\documentclass[12pt]{article}

\usepackage[utf8]{inputenc}
\usepackage[T1]{fontenc}

\usepackage{a4}
\usepackage{euscript}
\usepackage{amsmath}
\usepackage{amssymb}

\newcommand{\f}{\mathbb F}
\newcommand{\h}{\mathbf h}
\newcommand{\x}{\mathbf x}
\newcommand{\y}{\mathbf y}
\newcommand{\s}{\mathbf s}

\newcommand{\ALPHA}{\boldsymbol \alpha}
\newcommand{\BETA}{\boldsymbol \beta}
\newcommand{\HH}{\mathbf H}

\newcommand{\C}{\EuScript C}

\newcommand{\un}{\boldsymbol 1}
\newcommand{\supp}{\rm supp}
\newcommand{\NP}{\rm NP}
\newcommand{\RP}{\rm RP}
\newcommand{\ZPP}{\rm ZPP}
\newcommand{\coRP}{\rm coRP}
\def\rank#1{{\rm rank}{ (#1)}}

\newtheorem{lemma}{Lemma}
\newtheorem{definition}[lemma]{Definition}
\newtheorem{corollary}[lemma]{Corollary}
\newtheorem{theorem}[lemma]{Theorem}
\newtheorem{prop}[lemma]{Proposition}

\newenvironment{proof}{\noindent {\it Proof.~~}\ }{\  $\square$\medskip}

\title{On the hardness of the decoding and the minimum distance problems for rank codes}

\author{Philippe Gaborit\thanks{XLIM, Universit\'e de Limoges,
123, Av. Albert Thomas,
 87000 Limoges, France. {\tt gaborit@unilim.fr}}
\and
Gilles Z\'emor\thanks{Universit\'e de Bordeaux , Institut de Math\'ematiques
de Bordeaux, 351 cours de la Lib\'eration, 33405 Talence.
{\tt
zemor@math.u-bordeaux1.fr} }
}

\begin{document}

\maketitle

\begin{abstract}
In this paper we give a randomized reduction for the Rank Syndrome Decoding problem
and Rank Minimum Distance problem for rank codes. Our results are based on an embedding
from linear codes equipped with Hamming distance unto linear codes over an extension
field equipped with the rank metric.
We prove that if both previous problems for rank metric
are in ZPP = RP$\cap$coRP, then we would have NP=ZPP. 
We also give complexity results for the respective approximation problems in rank metric. 
\end{abstract}
\section{Introduction}

\subsection{General presentation}
The syndrome decoding problem for Hamming distance is a fundamental problem in complexity theory,
which gave rise to many papers over more than 30 years, since the seminal paper
of Berlekamp, McEliece and van Tilborg \cite{BMT78}, who proved the NP-completeness of the problem.
The problem of decoding codes is of first importance regarding applications, in particular 
for information theory and also for its connections with lattices.

Besides the notion of Hamming distance for error-correcting codes and the notion of Euclidean distance
for lattices, the concept of rank metric was introduced in 1951 by Loo-Keng Hua \cite{Loo51}
as "arithmetic distance" for matrices over a field $\f_q$.
Given two $n \times n$ matrices $A$ and $B$ over a finite field $\f_q$, the rank distance
between $A$ and $B$ is defined as $d_R(A,B)=Rank(A-B)$. In 1978, Delsarte defined in \cite{Del78}
the notion 
of rank distance on a set of bilinear form (which can also be seen as the set of rectangular matrices)and proposed a construction of optimal ${\it matrix \; codes}$ in bilinear form representation.
A matrix code over $\f_q$ for the rank metric is defined as the set of 
$\f_q$-linear combinations of a set $\mathcal{M}$ of
$m \times n$ matrices over $\f_q$.
Such codes are linear over $\f_q$ and the number $k$ of independent matrices in $\mathcal{M}$,
is bounded from above by $nm$. 
Then in 1985, Gabidulin  introduced in \cite{Gab85} 
the notion of rank codes in {\it vector representation} (as opposed to
{\em matrix} representation) over
an extension field $\f_Q$ of $\f_q$ (for $Q=q^m$). A rank code $C[n,k]$ of length $n$ and
dimension $k$ over $\f_Q$ in vector representation
is defined as a subspace over $\f_Q$ of dimension $k$ of $\f_Q^n$.   
It is possible to associate to any vector $x$ of $\f_Q^n$ an $m \times n$ matrix $X$ over $\f_q$
in the following way: let $x=(x_1,\cdots,x_n)$  in $\f_Q^n$
and let $\mathcal{B}$ be a basis of $\f_Q$ over $\f_q$. One can write any $x_i$ 
of the extension field $\f_Q$,
in the $\f_q$-linear basis $\mathcal{B}$, as a column vector $(x_{i1},\cdots,x_{im})^t$ 
of $\f_q^m$, so that one can associate an $m \times n$ matrix $X$ over $\f_q$
to any $x$ in $\f_Q^n$. The rank weight of $x$ is then defined as $w_R(x)=rank(X)$ 
and the rank distance
between $x$ and $y$ in $\f_Q^n$, is defined as $d_R(x,y)=rank(X-Y)$. 
Rank codes in vector representation can be seen as classical error-correcting codes over $\f_Q$ 
but embedded with the rank metric rather 
than with the Hamming metric, and 
one can define standard notions like generator and parity check matrices. 
Naturally any rank code $C[n,k]$ in vector representation is $\f_Q$-linear and can be seen
as a matrix code defined with $k \times m$ matrices over $\f_q$, but the converse is not true
and any rank matrix code has not, in general, a vector representation.
The vector representation is interesting because such codes are more compact to describe
and to handle.  In the following we will simply denote by {\it rank code}, 
a rank code in vector representation.

In 1985 \cite{Gab85},
 Gabidulin introduced an optimal class of rank codes (in vector representation): the 
so-called Gabidulin codes, which are evaluation codes, analogous to Reed-Solomon codes
but in a rank metric context, where monomial of the form $x^p$ are replaced by linearized
monomial of the form $x^{q^p}$ introduced by Ore in 1933 in \cite{Ore33}.

By analogy with the Hamming distance it is possible to define the two following problems:

{\bf Rank Syndrome Decoding problem (RSD)}\\
\begin{tabular}{ll}
  {\em Instance:} & a $(n-k)\times n$ matrix $H$ over $\f_Q^n$, a syndrome $s$ in $\f_Q^{n-k}$ and 
an integer $w$\\
  {\em Question:} & does there exist $x \in \f_Q^n$ such that
  $H.x^t=s$ and $w_R(x)\leq w$~?
\end{tabular}

and  

{\bf Rank Minimum Distance Problem (RMD)}\\
\begin{tabular}{ll}
  {\em Instance:} & a rank code $C[n,k]$, an integer $w$, \\
  {\em Question:} & does there exist $\x\in C$ such that
  $w_R(x)\leq w$~?
\end{tabular}

{\bf Remark:} the two previous problems fundamentally differ from the so-called
MinRank problem, which is also related to the rank metric but in a more general
case as it is explained in the next section.

The purpose of this paper is to study the computational complexity of the RSD problem
and the RMD problem, our main result reads as follows:

\begin{theorem}
  If the Rank Minimum Distance Problem for rank codes is in
  $\ZPP = \RP \cap \coRP$, then we must have $\NP=\ZPP$. Similarly, if the
  Rank Syndrome Decoding Problem for rank codes is in $\ZPP$, we must
  have $\NP=\ZPP$.
\end{theorem}

\subsection{ Previous work}
Surprisingly the theoretical computational complexities of the RSD and
RMD problems for the
rank metric are not known,
whereas the problem (and its variations) has been intensively studied for Hamming distance
or for lattices.
In particular besides the NP-completeness of the syndrome decoding problem for Hamming distance
proven in \cite{Var97}, the minimum distance problem for Hamming distance has been
proven NP-complete by Vardy in \cite{Var97}, as are also variations on the problem
\cite{CW09}.

As explained earlier in this introduction it is possible 
to consider the decoding and minimum distance problems in the rank
metric, but for {\em matrix} codes.
These
problems can be seen as generalizations of the RSD and RMD problems.
For instance for the case of the decoding problem for rank matrix codes,
we are given a set $\mathcal{M}=\{M_1,\cdots,M_k\}$ of $n \times n$
matrices over $\f_q$, a matrix $M$ over $\f_q$, and an integer $w$. The question
is to decide whether there exists an $\f_q$-matrix $M_0$ of rank $\leq
w$ such that $M-M_0$ can be expressed as an $\f_q$-linear combination
of matrices of $\mathcal{M}$ (i.e. is in the $\f_q$-linear matrix code
generated by the matrices of $\mathcal{M}$).
Note that we have linearity over
the small field $\f_q$ for the code, but not necessarily over the extension field $\f_Q$.
The latter decoding problem and its minimum distance variant  have appeared,
in slightly generalized forms, 
somewhat confusingly both under the name of ``MinRank'' in the literature.
Courtois
makes the observation in \cite{Cou01} that both the above problems for rank codes
in matrix representation are
NP-complete, by remarking that a Hamming metric code in $\f_q^n$ can
be ``lifted'' into a rank metric code in matrix representation simply by transforming any
vector $x$ of $\f_q^n$ into a diagonal matrix with $x$ written on
the diagonal. By this process a Hamming code of dimension $k$ with minimum distance $d$
is
lifted unto a rank-metric code in matrix representation with $\mathcal{M}$
a set with $k$ matrix element, with rank minimum distance $d$. This
transformation yield the NP-completeness of the previous decoding problem for matrix codes
from the NP-completeness of Syndrome Decoding problem for
classical Hamming codes. The NP-completeness of the associated Minimum Rank Distance problem
follows similarly from the NP-completeness of minimum distance problem for Hamming metric.

However, in the case of matrix representations the structure of the linear matrices
over $\f_q$ is simpler than the structure for rank $[n,k]$ codes in vector representation
which are linear
over the extension field $\f_Q$ and not only on the base field $\f_q$.
The "MinRank" problem
appears as a weakly-structured variation of the RSD and RMD problems.
The above remark by Courtois
works well for $\f_q$-linear matrix codes but clearly does not apply
for $\f_Q$-linear $[n,k]$ rank codes.

\subsection{Applications of the rank metric}

Over the years the notion of rank metric has become a very central tool for new applications of coding
theory and has also very interesting applications to cryptography.

{\bf Applications to coding theory.} Concerning coding theory, from the end of '90s new application contexts appeared for coding theory: 
space-time coding \cite{TSC98} in 1997 and network coding in 2001 \cite{LYC03}.

Space time coding was introduced by Tarokh, Jafarkhani and Calderbank
in 1998 in \cite{TSC98}. One strives to improve the reliability of
data transmission in wireless communication systems using multiple
transmission antennas. This redundancy results in a higher chance of
being able to use one or more of the received copies to correctly
decode the received signal. In fact, space–time coding combines all
the copies of the received signal in an optimal way to extract as much
information from each of them as possible. A full rank criterion was
proposed for choosing rank matrices  with full rank difference, which
enables one to decode errors in this context.


For network coding introduced in 2001 in \cite{LYC03}, the idea is optimize information sent 
in given time slots, when the information is sent from a single source to a single destination
through a network with nodes which send random linear combination of received information.
Koetter and Kschichang introduced in 2007 \cite{KK08} the notion of subspace metric (which
is a small variation on the rank metric \cite{GP08}),
and the so-called Koetter-Kschichang codes which are an adaptation of the
Gabidulin codes in a subspace metric context. 

More generally a lot of work has also been done for decoding Gabidulin or Koetter-Kschichang
codes, though admittedly somewhat less than for Reed-Solomon codes,
their Hamming distance counterparts: in particular list-decoding algorithms
are known only for subclasses of Gabidulin codes and not yet for the general family of codes \cite{GX13,GNW12,MV13}.

{\bf Applications to cryptography} 

Rank-based cryptography belongs to the larger class of post-quantum cryptosystems,
which is an alternative class of cryptosystems which are {\it a priori} resistant to a putative
quantum computer. The first cryptosystem was proposed in  
1991 by Gabidulin, Paramonov and Tretjakov 
(the GPT cryptosystem \cite{GPT91} which adapts the McEliece
cryptosystem to the rank metric and Gabidulin codes).





The particular interest of rank metric based problems compared to lattices or (Hamming) codes based
problems is that the practical complexity of the best known attacks for rank-based problems \cite{GRS12} grows  very quickly
compared to their Hamming counterpart \cite{Jou12}.
Indeed such attacks have a quadratic term (related to parameters of the rank code)
in their exponential coefficient,
 while for Hamming distance problems ( and somehow also for heuristic LLL attacks for lattice-based cryptographic) ,
 the best practical attacks 
have only an exponential term whose exponent is linear
in the code parameters. This translates into rank codes having
a decoding complexity that behaves as $exp(\Omega(N^{2/3}))$ rather than
$exp(\Omega(N^{1/2}))$
for Hamming codes, where $N$ is the input size, i.e. the number of $q$-ary symbols needed to
describe the code.




In practice it means that it is possible to obtain secure practical parameters for random
instances in rank metric 
of only a few thousand bits related to a hard problem, when at least a
hundred thousand bits are 
needed for Hamming distance or for lattices.
Such random instances for rank metric are used for instance, for zero-knowledge authentication 
in \cite{GSZ11}, and weakly structured instances are used for the recent LRPC cryptosystem \cite{GMRZ13} 
(similar to the NTRU cryptosystem \cite{ntru} for lattices and the recent MDPC cryptosystem for codes)
or for signature \cite{ranksign}.
Of course with (Hamming) codes and lattices it is possible
to decrease the size of parameters to a few thousand bits 
with additional structure \cite{BCGO09,LPR13},
but then the reduction properties to hard problems are lost because they are reduced to decoding problems for
  special classes of codes whose complexity is not known..

Overall because of the practical complexity of best known attacks, rank-based cryptography
has very good potential for cryptography, furthermore, our
    present results show that cryptographic protocols whose security
    can be reduced to the decoding problem for rank codes will have
    both reduction to a proven hard problem and the potential for
    small keys. Finally, we remark that since the codes actually used
    for rank-metric applications, cryptographic or otherwise, tend to
    be rank-codes in the sense of this paper, i.e. with linearity over
    the large field, the decoding and minimum distance problems for
    these codes are more relevant than the same problems for the
    looser matrix code class, whose NP-completeness has been referred
    to a number of times in the past.

{\bf Organization of the paper:} the paper is organized as follows, in Section 2, we give an overview
of our results and describe our embedding technique, in Section 3 we give  
a probabilistic analysis of our reduction setting, Section 4 describes our main results,
and finally Section 5 considers further results on approximation
problems for the rank metric.

\section{Overview}

It is clear that Courtois's diagonal embedding of the Hamming space
into the rank metric space works well for rank codes in matrix form linear over $\f_q$
but does not work for rank codes with linearity over the extension field $\f_Q$.
We shall therefore introduce a different embedding strategy defined as follows:

\begin{definition}
  Let $m\geq n$ and $Q=q^m$. Let $\ALPHA = (\alpha_1,\ldots \alpha_n)$
  be an $n$-tuple of elements of $\f_Q$. Define the embedding of
  $\f_q^n$ into $\f_Q^n$ 
  $$\begin{array}{lcll}
       \psi_{\ALPHA}: & \f_q^n & \rightarrow & \f_Q^n\\
             & x=(x_1,\ldots ,x_n) & \mapsto & \x =(x_1\alpha_1,\ldots x_n\alpha_n)
  \end{array}$$
  and for any $\f_q$-linear code $C$ in $\f_q^n$, define
  $\C=\C(C,\ALPHA)$ as the $\f_Q$-linear code generated by
  $\psi_{\ALPHA}(C)$, i.e. the set of $\f_Q$-linear combinations of
  elements of $\psi_{\ALPHA}(C)$.
\end{definition}

{\bf Remark:} The condition $m\geq n$ ensures that, by adjoining $m-n$
zeros to vectors of $\f_Q^n$, they may be seen as $m\times m$ matrices
so that the code $\C$ may be seen as a rank code.

It should be clear that for any $\ALPHA$, the rank weight of
$\psi_{\ALPHA}(x)$ is at most the Hamming weight of $x$: therefore the
Minimum Rank Distance of $\C$ never exceeds the Hamming minimum
distance of the original code $C$. It may however be less. For
example, if $\ALPHA = \un =(1,1,\ldots ,1)$ then the rank weight of
$\psi_{\ALPHA}(x)$ is always $1$ for every $x\neq 0$. The minimum rank weight
of $\C(C,\ALPHA)$ may also be less than the minimum Hamming distance
of $C$ for more sophisticated reasons. In particular, if
$\alpha_1,\ldots \alpha_n$ are $\f_q$-linearly independent, we have
that the rank weight of $\psi_{\ALPHA}(x)$ is always equal to the
Hamming weight of $x$, but the minimum rank weight of $\C$ may still
be less than the minimum Hamming distance of $C$. Consider for
instance the binary code $C$ of words of even weight of length $3$, we
have that 
$$\x = \alpha_2\psi_{\ALPHA}(101)+\alpha_1\psi_{\ALPHA}(011) = 
(\alpha_1\alpha_2,\alpha_1\alpha_2,\alpha_3(\alpha_1+\alpha_2)).$$
Now if $\alpha_3$ happens to have been chosen equal to
$\alpha_1\alpha_2(\alpha_1+\alpha_2)^{-1}$, we will have that
$\rank{x}=1<d_{\rm Hmin}(C)$ even though $\ALPHA$ may very well be of
rank $3$.

If, given any Hamming code $C$, we could efficiently find an $n$-tuple
$\ALPHA$ that would guarantee that $C(C,\ALPHA)$ has minimum rank
distance equal to $d_{\rm Hmin}(C)$, we would have a polynomial
reduction that would derive the NP-completeness of the Minimum Rank
Distance problem for rank codes from the NP-completeness of
the classical minimum Hamming distance problem. We have not been able
to see how to do this in any deterministic way. However, we shall show
that when $\ALPHA$ is chosen {\em randomly}, for $m=O(n)$, then we
probability tending to $1$ we have $d_{\rm Rmin}(\C(C,\alpha)) =
d_{\rm Hmin}(C)$. This makes the Rank Minimum Distance hard for NP
under unfaithful random reductions (UR reductions, in the terminology of
\cite{Jo90}). As a consequence we have that if the Rank Minimum
Distance Problem were in coRP we would have NP $\subset$ coRP. With a further
transformation we shall obtain that if the Rank Minimum
Distance Problem were in RP then we would have also NP $\subset$ RP:
our results will therefore show that if the Rank Minimum
Distance Problem were in ZPP = RP$\cap$coRP, then we would have
NP=ZPP.

\section{Probabilistic analysis of our embedding }

\subsection{Notation and definitions}

We refer to \cite{MS77} and \cite{Loi06} for general results on codes
and rank codes.
Let $\f_q$ be a field with $q$ elements and let $\f_Q$, with $Q=q^m$, be an extension of $\f_q$
of degree $m$. In the following we consider two type of codes, codes with Hamming distance
considered as $C[n,k,d_H]$ linear codes over the base field $\f_q$, for $n$ and $k$ the length
and dimension of the code and $d_H$, its minimum Hamming distance. we also consider
 rank codes with rank distance written as $C[n,k,d_R]$ linear codes over the field $\f_Q$ of length
$n$, dimension $k$ an minimum rank distance $d_R$, embedded
with the rank metric.

We recall the Griesmer bound for linear codes over $\f_q$ that ill be useful for our proofs:

\begin{prop}[Griesmer bound]
Let $C$ be a $[n,k,d]$ over $\f_q$ then
$$ n \ge \sum_{i=0}^{k-1} \lceil \frac{d}{q^i} \rceil$$
\end{prop}

\subsection{Probabilistic analysis of $\C(C,\ALPHA)$}

\begin{lemma}\label{lem:dual}
  Let $C^\perp$ be the dual code of $C$ over $\f_q$. Let $\BETA =
  \ALPHA^{-1} = (\alpha_1^{-1},\ldots ,\alpha_n^{-1})$. Then
  $\C(C^\perp,\BETA)$ is the dual code of $\C(C,\ALPHA)$ over $\f_Q$
  and $\dim_{\f_Q}\C(C,\ALPHA) = \dim_{\f_q} C$.
\end{lemma}

\begin{proof}
  It should be clear that $\C(C,\ALPHA)$ and $\C(C^\perp,\BETA)$ are
  orthogonal to each other. Choosing systematic generator matrices for
  $C$ and $C^\perp$ shows that $\dim_{\f_Q}\C(C,\ALPHA) = \dim_{\f_q}
  C$ and $\dim_{\f_Q}\C(C^\perp,\BETA) = \dim_{\f_q}C^\perp$.
\end{proof}

\begin{lemma}\label{lem:support}
  Let $C$ be an $\f_q$-linear code of $\f_q^n$ and let
  $W\subset\{1,2,\ldots , n\}$ be a set of coordinates such that no
  non-zero codeword of $C$ has its Hamming support included in $W$. Then, for
  any $j\in W$ and for any $\ALPHA$, there is a codeword $\x$ of
  $\C(C,\ALPHA)^\perp$ such that $\supp(\x)\cap W=\{j\}$.
\end{lemma}

\begin{proof}
  Since the code $C$, punctured so as to leave only the coordinates in
  $W$, has only the zero codeword, we have that for any $j\in W$ there
  is $x$ in $C^\perp$ such that $\supp(x)\cap W=\{j\}$. The conclusion
  follows by Lemma \ref{lem:dual}.
\end{proof}

\begin{corollary}\label{cor:d}
  Let $C$ be an $\f_q$-linear code of $\f_q^n$ with minimum Hamming
  distance $d$. Then, for any $\ALPHA$, the {\em Hamming} minimum
  distance of the embedded code $\C(C,\ALPHA)$ is equal to $d$.
\end{corollary}

\begin{proof}
  That it is at most $d$ is clear by the definition of
  $\C(C,\ALPHA)$. To see that it is at least $d$ follows from
  Lemma~\ref{lem:support}. 
\end{proof}

\begin{lemma}\label{lem:(q+1)/q}
  Let $C$ be an $\f_q$-linear code of $\f_q^n$ with minimum Hamming
  distance $d$. Let $w<\frac{q+1}{q}d$. Then, for any
  $\ALPHA=(\alpha_1,\ldots ,\alpha_n)$, the only codewords of
  $\C(C,\ALPHA)$ of Hamming weight $w$ are of the form
  $\lambda\psi_{\ALPHA}(x)$, $\lambda\in\f_Q$, for some codeword of $C$. In particular, if
  $\alpha_1,\ldots ,\alpha_n$ are linearly independent, then any
  codeword of $\C(C,\ALPHA)$ of Hamming weight $w$ is also of rank weight $w$.
\end{lemma}

\begin{proof}
  Let $W\subset\{1,2,\ldots , n\}$ be a set of $w$ coordinates. Let
  $C_{|S}$ be the corresponding shortened code of $C$, i.e. the set of
  codewords of $C$ of support included in $W$. By
  Lemma~\ref{lem:dual}, we have that the dual code of
  $\C(C,\ALPHA)_{|S}$ has dimension $w-\dim C_{|S}$ and therefore 
  $\dim\C(C,\ALPHA)_{|S} = \dim C_{|S} = \dim\C(C_{|S},\ALPHA)$. By the Griesmer
  bound, the dimension of $C_{|S}$ is at most $1$. Therefore the only
  codewords of $\C(C_{|S},\ALPHA)$ are of the form $\lambda\psi_{\ALPHA}(x)$.
\end{proof}

\begin{theorem}\label{thm:distance}
  Subject to the condition $m>2qn$, when $\ALPHA$ is chosen randomly
  and uniformly in $\f_Q^n$, then for any linear code $C\in\f_q^n$, 
  the probability that the rank
  minimum distance of $\C(C,\ALPHA)$ differs from the Hamming minimum
  distance of $C$ is bounded from above by a quantity that vanishes
  exponentially fast in $n$. 
\end{theorem}

\begin{proof}
  Let $C$ be fixed and let $d$ be its Hamming minimum distance. 
  It suffices to prove that for any Hamming weight
  $w\leq n$, the probability $P_w$ that there exists a codeword of
  $\C(C,\ALPHA)$ of Hamming weight $w$ and of rank weight $<d$
  vanishes exponentially fast.
  \begin{itemize}
  \item $w< d+d/q$. If $w<d$, Then by Corollary~\ref{cor:d},
    $P_w=0$. Otherwise, by Lemma~\ref{lem:(q+1)/q}, $P_w$ is bounded
    from above by the probability that $\alpha_1\ldots \alpha_n$ are
    linearly dependent, which is exponentially small in $n$.
  \item $w\geq d+d/q$. We bound from above $P_w$ by the expected
    number of codewords of $\C(C,\ALPHA)$ of rank weight $<d$ and
    Hamming weight $w$. Let $\x$ be a vector of $\f_Q^n$ of Hamming
    weight $w$ and let $W$ be the Hamming support of $\x$, so that
    $w=|W|$. Let $J$ be a maximal subset of coordinates of $W$ such
    that no nonzero codeword of $C$ has its support included in
    $I$. By Lemma~\ref{lem:support}, we have that there are $|J|$
    parity-checks for the event $\x\in\C(C,\ALPHA)$ that are satisfied
    each with probability $1/Q$ and, truncated to $W$, are independent over $\f_Q$
    and therefore are satisfied independently in the sense of probability.
    Hence, the probability that $\x$ is a codeword of $\C(C,\ALPHA)$
    is at most $1/Q^{|J|}$. By the Griesmer bound, we have $|J|\geq d+d/q-1$.
    Bounding from above the number $N_w$ of vectors of $\f_Q^n$ of
    Hamming weight $w$ and rank weight $d-1$ by:
    $$N_w\leq \binom{n}{w}Q^{d-1}(q^{d-1})^w\leq
    2^nq^{m(d-1)+w(d-1)}$$
we obtain
\begin{align*}
  P_w&\leq N_w\frac 1{q^{md+md/q-m}}\\
     &\leq 2^nq^{w(d-1)-md/q}
\end{align*}
and the result follows from the hypothesis $m/q>2n$.
  \end{itemize}
\end{proof}

\section{The syndrome decoding problem}

Let us recall the syndrome decoding problem:

\begin{tabular}{ll}
  {\em Instance:} &
  \begin{minipage}[t]{0.8\linewidth}
    an $r\times n$ matrix $\HH=[\h_1,\h_2,\ldots
,\h_n]$ over a field $\f$, a column vector $\s\in\f^r$, an integer $w$
  \end{minipage}\\
  {\em Question:} & \begin{minipage}[t]{0.8\linewidth}
  does there exist $\x=(x_1,\ldots ,x_n)\in \f^n$ of
  weight at most $w$ such that
  $\sigma(\x)=\HH\,^t\!\x =\sum_{i=1}^nx_i\h_i=\s$~?
  \end{minipage}
\end{tabular}

When $\f=\f_q$ and the weight refers to the Hamming weight,
we have the classical or {\em Hamming} syndrome
decoding problem: when $\f=\f_Q$ and the weight refers to the rank (or
rank weight) we have the {\em rank} syndrome
decoding problem. It is classical that the syndrome decoding problem
is equivalent to the decoding (or closest vector) problem, because
looking for the closest codeword to a given vector $\y$ amounts to
computing the syndrome $\s=\sigma(\y)$ of $\y$ and solving the syndrome decoding
problem for $\s$ (subtracting the solution to $\y$ gives the closest
codeword).

Since the Hamming syndrome decoding problem is known to be
NP-complete, it is natural to try and devise a transformation from it
to the rank syndrome decoding problem. For this purpose, let us
introduce the following notation:
for any $r\times n$ matrix $\HH=[\h_1,\h_2,\ldots
,\h_n]$ of elements of $\f_q$, and for any $\BETA=(\beta_1,\ldots
,\beta_n)$, $\beta_i\in\f_Q$, denote by $\HH(\BETA)$ the matrix 
$$\HH(\BETA) = [\beta_1\h_1,\beta_2,\h_2,\ldots ,\beta_n\h_n].$$

Our strategy is, given an instance $(\HH,\s,w)$ of the Hamming
syndrome decoding problem, to associate to it the transformed instance
$(\HH(\BETA),\s,w)$ of the rank decoding problem. It is clear that if
$x$ is a solution to the Hamming syndrome decoding problem, then 
$\x=\psi_{\ALPHA}(x)$ is a solution to the associated rank syndrome
decoding problem with
$\ALPHA=\BETA^{-1}=(\beta_1^{-1},\ldots,\beta_n^{-1})$, the rank
weight of $\x$ being not more than the Hamming weight of $x$. Again,
we strive to show that when choosing $\BETA$ randomly and uniformly,
the smallest rank weight of a solution to $(\HH(\BETA),\s,w)$ is very
probably equal to the smallest Hamming weight of a solution to
$(\HH,\s,w)$.

\begin{lemma}\label{lem:w_H}
  Let $\HH$ be an $r\times n$ matrix and let $\s$ be a column vector
  of $\f_q^r$. Let $w_H$ be the minimum Hamming weight of a vector $x$
  of $\f_q^n$ of syndrome $\sigma(x)=\HH\,^t\!x = \s$. Let
  $\x\in\f_Q^n$ be such that $\sigma_{\BETA}(\x) = \HH(\BETA)\,^t\!\x
  = \s$. Then, if $J\subset\{1,\ldots ,n\}$ is the Hamming support of
  $\x$, there exists a subset $W\subset J$ such that $|W|=w_H$ and the
  columns 
  $(\h_j)_{j\in W}$ of $\HH$ indexed by $W$ are $\f_Q$-linearly independent.
\end{lemma}

\begin{proof}
  Let $W$ be a maximal subset of the support of $\x$ such
  $(\beta_j\h_j)_{j\in W}$ is $\f_Q$-linearly independent. Since
  $\sigma_{\BETA}(\x) =\s$, we must also have that $\s$ belongs to the
  $\f_Q$-linear span of $(\beta_j\h_j)_{j\in W}$. Now by
  Lemma~\ref{lem:dual} we have that $\f_Q$-linear independence of
  $(\beta_j\h_j)_{j\in I}$ (and therefore also simply of $(\h_j)_{j\in
    I}$)
  is equivalent to $\f_q$-linear independence of $(\h_j)_{j\in I}$ for
  any set $I$ of coordinates. Since any set of columns of $\HH$
  that generate $\f_q$-linearly $\s$ must be of size at least $w_H$ by
  definition of $w_H$ we have $|W|\geq w_H$.
\end{proof}

\begin{theorem}\label{thm:decoding}
  Subject to the condition $m>n^2$, when $\BETA$ is chosen randomly and
  uniformly in $\f_Q$, then for any $r\times n$ matrix $\HH$ over
  $\f_q$ and any column vector $\s\in\f_q^r$, denoting by $w_H$ the
  minimum Hamming weight of a vector of $\f_q^n$ of syndrome $\s$ by
  $\HH$ and by $w_R$ the minimum rank weight of a vector of $\f_Q^n$
  of syndrome $\s$ by $\HH(\BETA)$, we have that the probability that
  $w_H\neq w_R$ is bounded from above by a quantity that vanishes
  exponentially fast in $n$.
\end{theorem}
\begin{proof}
  Let $\HH$, $\s$ and $w_H$ be fixed. Since we have remarked that
  $w_R\leq w_H$, It suffices to show for every integer $w\leq n$ that
  the probability $P_w$ that there exists a codeword of $\f_Q^n$ of
  syndrome $\s$ by $\HH(\BETA)$ and of Hamming weight $w$ and rank
  weight $<w_H$, is a quantity that vanishes exponentially fast with
  $n$.

  By Lemma~\ref{lem:w_H}, if $w<w_H$ we have $P_w=0$. Suppose
  therefore $w\geq w_H$. Let $\x$ be a vector of $\f_Q^n$ of Hamming
  weight $w$. 
  Lemma~\ref{lem:w_H} implies that there are at
  least $w_H$ columns of $\HH$ indexed by nonzero coordinates
  of $\x$ that are $\f_Q$-linearly independent. This implies that the
  span of $\f_Q$-linear combinations of these $w_H$ columns has size
  $Q^{w_H}$, and therefore the probability that the syndrome by
  $\HH(\BETA)$ of $\x$ equals $\s$ is at most $1/Q^{w_H}$.

  Bounding from above $P_w$ by the expectation of the number of
  codewords of Hamming weight $w$ and rank weight $\leq w_H-1$, we
  have:
  \begin{align*}
    P_w &\leq \binom{n}{w}Q^{w_H-1}(q^{w_H-1})^w\frac{1}{Q^{w_H}}\\
        &\leq 2^nq^{n(w_H-1)}\frac 1Q\\
        &\leq q^{nw_H-m}
  \end{align*}
which proves the result since the case $w_H=n$ is easily dealt with
separately. 
\end{proof}

Theorems \ref{thm:distance} and \ref{thm:decoding} yield Theorem 1, our main result
stated in the introduction:

{\bf Proof of Theorem 1}

\begin{proof}
  That $\NP\subset \coRP$ follows directly from the NP-completeness of
  the Hamming Minimum Distance Problem and Theorem~\ref{thm:distance}
  in the first case and from the NP-completeness of the Hamming
  Syndrome Decoding Problem and Theorem~\ref{thm:decoding} in the
  second case: the original Hamming problem is simply transformed
  by a probabilistic embedding into the corresponding Rank metric
  problem. To be precise, the hypothesis that the Rank Minimum
  Distance Problem is in $\coRP$ means that there is probabilistic
  polynomial time algorithm that always outputs ``yes'' on ``yes''
  instances and often outputs ``no'' on ``no'' instances. Applied to
  a code $\C(C,\ALPHA)$ for a random $\ALPHA$, we obtain an algorithm
  that, for the original Hamming Minimum Distance Problem
  always outputs ``yes'' (the minimum distance is not more than a
  given value) on ``yes'' instances and often ``no'' otherwise.

  Next we deduce from the hypothesis that the Rank
  Minimum Distance Problem for rank codes is in RP that
  $\NP\subset \RP$.
  We need to construct a probabilistic
  algorithm that given an integer $w$ and a Hamming code with minimum
  distance $d_W>w$ always decides ``no'', and often decides
  ``yes'' when the minimum distance $d_H$ is not more than $w$. To achieve
  this we find a witness for $d_H\leq w$. The hypothesis that the Rank
  Minimum Distance Problem is in $\RP$ means that there is a
  probabilistic polynomial time algorithm that always decides ``no''
  when the rank minimum distance $d_R$ is above $w$ and very often decides
  ``yes'' when it $d_R\leq w$. Suppose that the Hamming code $C$ is
  such that $d_H\leq w$. We transform $C$ into a random $\C(C,\ALPHA)$
  and ask the probabilistic machine for the Minimum Rank Distance
  whether $d_R\leq w$. If the answer is ``no'' we output a ``no''. If
  it is ``yes'', we remove the first column from a fixed parity-check matrix
  of $C$ and start the procedure (create another random rank-metric
  code from the shortened version of $C$) again. If the answer is
  ``no'', we put back the removed column and start again by removing
  the second column, until we either run out of columns to remove in
  which case we output a final ``no'', or we obtain a ``yes'', in
  which case we continue removing columns, always of a larger index
  than the columns we have previously tried to remove. We stop
  removing columns if we reach a point when only $w$ columns
  remain. At this point we check that the thus shortened Hamming code
  has dimension at least $1$, in which case we ``output'' a
  ``yes''. In all other cases we output a ``no''.

  We see that the number of times we use randomness and access the
  rank minimum distance oracle is at most $n$. Furthermore, if it is
  true that $d_H\leq w$ for the original code, then with probability
  exponentially close to $1$ for large $n$ and fixed $q$ we will
  obtain a ``yes'', and if it is not true that $d_H\leq w$ we will
  always obtain a ``no''. This is a random polynomial time algorithm
  that puts an NP-complete problem (Minimum Distance for Hamming
  linear codes) in RP, hence the result.

  To reach the same conclusion from the hypothesis that the Rank
  Syndrome Decoding Problem for rank codes is in $\RP$ we
  use a very similar witness constructing technique for the Hamming
  syndrome decoding problem.
\end{proof}

{\bf Remark.} The reduction is somewhat looser in the Decoding case
where an extension field of quadratic degree in $n$ is needed, 
than in the Minimum distance case where a degree linear in $n$ was
sufficient. This is somewhat surprising, since in the more well-known
Hamming distance and Lattice cases, the Minimum Distance problem has
been more difficult to reduce than the Decoding problem.

\section{Further results on approximation problems for rank metric}

The syndrome decoding problem and the minimum distance problem for Hamming distance
are connected to other interesting problems. It is natural to consider generalizations
of these problems from the Hamming distance to rank metric, especially with the
use of our very versatile embedding. In the following as an example of application 
of our embedding we consider the case
of two particular well known approximation problems in Hamming distance: 
the Gap Minimum Distance Problem (GapMDP), for which we want
to approximate the minimum distance of a code up to a constant
and Gap Nearest Codeword (GapNCP) in which we want to approximate the decoding 
distance.  Notice that equivalently the previous (GapNCP) problem can be stated in terms 
of Syndrome Decoding with a parity check matrix.  

These approximation problems are stated in the following way:
 
\begin{definition}[$GapMDP_{q,\gamma}$]
For a prime power $q$ and $\gamma \ge 1$, an instance
of the Gap Minimum Distance problem $GapMDP_{q,\gamma}$ is a linear code $C$ over
$\f_q$, given by its generator matrix, and an integer $t$ such that:

$\bullet$ it is a YES instance if $d_H(C) \le t$;

$\bullet$ it is a NO instance if $d_H(C) > \gamma t$
\end{definition}

\begin{definition}[$GapNCP_{q,\gamma}$]
For a prime power $q$ and $\gamma \ge 1$, an instance $(C,v,t)$
of the Gap Minimum Distance problem $GapNCP_{q,\gamma}$ is a linear code $C$ over
$\f_q$, given by its generator matrix, $v \in \f_q^n$ and a positive integer $t$.

$\bullet$ it is a YES instance if $d_H(v,C) \le t$;

$\bullet$ it is a NO instance if $d_H(v,C) > \gamma t$
\end{definition}

 Both these promise problems
have been proven NP-complete for Hamming distance for $\gamma > 1$
respectively in \cite{CW09} (see also \cite{DMM99}) and \cite{ABS97}.

The generalization of these problems to the rank metric is straightforward:
we may define Gap Rank Minimum Distance (GapRMPD) and Gap Rank Nearest Codeword Problem (GapRNCP):

\begin{definition}[$GapRMDP_{q,\gamma}$]
For a prime power $q$, an integer $m$ and $\gamma \ge 1$, an instance
of the Gap Rank Minimum Distance problem $GapMDP_{q,\gamma}$ is a linear rank code $C$ over
$F_{q^m}$, given by its generator matrix, and an integer $t$ such that:

$\bullet$ it is a YES instance if $d_R(C) \le t$;

$\bullet$ it is a NO instance if $d_R(C) > \gamma t$
\end{definition}

\begin{definition}[$GapRNCP_{q,\gamma}$]
For a prime power $q$, an integer $m$ and $\gamma \ge 1$, an instance $(C,v,t)$
of the Gap Rank Minimum Distance problem $GapRNCP_{q,\gamma}$ is a linear rank code $C$ over
$F_{q^m}$, given by its generator matrix, $v \in F_{q^m}^n$ and a positive integer $t$.

$\bullet$ it is a YES instance if $d_R(v,C) \le t$;

$\bullet$ it is a NO instance if $d_R(v,C) > \gamma t$
\end{definition}

We then deduce the following corollary:

\begin{corollary}
If the problems $GapRMDP_{q,\gamma}$ and $GapRNCP_{q,\gamma}$ are in $coRP$
then NP=ZPP.
\end{corollary}
\begin{proof}
We use the same embedding technique as for Theorem 1.
 Since the Hamming distance is always greater or equal than the rank distance,
we obtain a Unfaithful Random (UR) reduction between the respective approximation Hamming distance
problems and rank distance problems and hence by the result of (\cite{Jo90},p.118), the result
follows.
\end{proof}

\section{Conclusion}

In this paper we proved the hardness of the minimum distance and syndrome decoding problems
for rank codes and rank distance under a randomized UR reduction. If we compare to other type of metrics
like Hamming or Euclidean distance, we see that, for the decoding problem
the reductions for codes equipped with Hamming distance and lattices with Euclidean distance are deterministic
and for minimum distance, reductions are randomized for lattices and deterministic for codes (see \cite{Mic13}
and references therein).
A worthwhile challenge would be to obtain a deterministic reduction also for rank metric. 

\begin{center}
ACKNOWLEDGEMENT
\end{center}

The first author thanks O. Ruatta for helpful discussions.

\end{document}